\documentclass[11pt,a4paper,reqno,twoside]{article}

\usepackage{amssymb}
\usepackage{latexsym}
\usepackage{amsmath}
\usepackage{amsthm}
\usepackage{enumitem,cite}
\usepackage{tikz}
\usepackage{pgfplots}
\setlist[enumerate]{leftmargin=1.4cm}
\usepackage{bm}
\usepackage{subcaption}
\usepackage[margin=3.1cm]{geometry}

\definecolor{blueish}{rgb}{0.0, 0.53, 0.74}
\definecolor{blush}{rgb}{0.87, 0.36, 0.51}

\numberwithin{equation}{section}
\numberwithin{figure}{section}

\theoremstyle{definition}
\newtheorem{theorem}{Theorem}[section]
\newtheorem{corollary}[theorem]{Corollary}
\newtheorem{proposition}[theorem]{Proposition}
\newtheorem{definition}[theorem]{Definition}

\newtheorem{notation}[theorem]{Notation}
\newtheorem{remark}[theorem]{Remark}
\newtheorem{lemma}[theorem]{Lemma}

\newcommand\qbin[3]{\left[\begin{matrix} #1 \\ #2 \end{matrix} \right]_{#3}}
\newcommand\binomi[2]{\left(\begin{matrix} #1 \\ #2 \end{matrix} \right)}

\newcommand{\numberset}{\mathbb}

\newcommand{\F}{\numberset{F}}

\newcommand{\HH}{\textnormal{H}}
\newcommand{\II}{\textnormal{I}}

\newcommand{\mV}{\mathcal{V}}

\newcommand{\mC}{\mathcal{C}}

\newcommand{\mG}{\mathcal{G}}

\newcommand{\mF}{\mathcal{F}}
\newcommand{\mW}{\mathcal{W}}

\newcommand{\mE}{\mathcal{E}}

\newcommand{\mB}{\mathcal{B}}

\renewcommand{\longrightarrow}{\to}

\newcommand{\dH}{d^{\textnormal{H}}}

\newcommand{\dI}{d^{\textnormal{I}}}
\newcommand{\deltaI}{\delta^{\textnormal{I}}_q}
\newcommand{\deltaH}{\delta^{\textnormal{H}}_q}

\newcommand{\BallH}{\bm{b}^\HH}

\newcommand{\BallI}{\bm{b}^\textnormal{I}}

\newtheorem{claim}{Claim}

\newcommand*{\myproofname}{Proof of the claim}
\newenvironment{clproof}[1][\myproofname]{\begin{proof}[#1]}{\end{proof}}

\usepackage{hyperref}
\usepackage{graphicx}

\usepackage{setspace}
\usepackage{calrsfs}

\title{{The Typical Non-Linear Code over Large Alphabets}}

\usepackage{authblk}

\author{Anina Gruica$^*$ and Alberto Ravagnani\thanks{Anina Gruica is supported by the Dutch Research Council through grant OCENW.KLEIN.539. Alberto Ravagnani is supported by the Dutch Research Council through grants OCENW.KLEIN.539 and VI.Vidi.203.045.}}
\affil{Department of Mathematics and Computer Science \\ Eindhoven University of Technology, the Netherlands}

\date{}

\begin{document}

\maketitle

\begin{abstract}
We consider the problem of describing the typical (possibly) non-linear code of minimum distance bounded from below over a large alphabet. We concentrate on block codes with the Hamming metric and on subspace codes with the injection metric. In sharp contrast with the behavior of linear block codes, we show that the typical non-linear code in the Hamming metric of cardinality $q^{n-d+1}$ is far from having minimum distance $d$, i.e., from being MDS. We also give more precise results about the asymptotic proportion of block codes with good distance properties within the set of codes having a certain cardinality. We then establish the analogous results for subspace codes with the injection metric, showing also an application to the theory of partial spreads in finite geometry.
\end{abstract}

\section{Introduction}

Understanding how the typical error-correcting code having certain properties looks like is a standard problem in information theory. The most common question is probably whether or not a uniformly random  code meets a given bound (e.g. the Hamming or the Singleton bound) with equality when the block length goes to infinity.
This question was formally addressed in
\cite{barg2002random} using probability methods; see also~\cite{hao2020distribution} and~\cite{gallager1973random}, among many others, on closely related topics.


In this paper, in contrast with previous literature, we ask ourselves how the typical code over a large alphabet looks like, keeping the block length fixed. It is well-known and quite easy to see that most of the linear block codes having a certain dimension are MDS over a sufficiently large field. More precisely, for all $n \ge k \ge 1$ we have
$$\lim_{q \to +\infty} \frac{\mbox{\,number of $[n,k]_q$ MDS codes\,}}{\mbox{number of $[n,k]_q$ codes}}=1.$$
In other words, the probability that a uniformly random $k$-dimensional linear block code is MDS approaches 1 as the field size $q$ grows. Note that the same result is far from being true when the block length goes to infinity.

\begin{remark}
An intuitive explanation for the density of linear MDS codes might be that, for $q$ large, most pairs of vectors are far from each other in the Hamming metric. Accordingly, a uniformly random set of such vectors is expected to have optimal distance properties. In this short paper, we show that this intuitive explanation is in fact incorrect.
\end{remark}

We depart from the linear setting and investigate the distance properties of the typical (possibly) non-linear block code of a given cardinality, finding that the probability that a uniformly random code of cardinality $q^{n-d+1}$ has minimum distance~$d$ goes to 0 as $q$ grows. This is in sharp contrast with the behavior of linear MDS codes described above. In fact, we investigate more generally the asymptotic density of codes of given cardinality and minimum distance bounded from below, showing that the decisive cardinality for sparseness/density is (asymptotically) 
$$\sqrt{q^{n-d+1}},$$
i.e., the square root of the maximum cardinality of a block code
$\mC \subseteq \F_q^n$ of minimum distance lower bounded by $d$.

The proof techniques used in this paper rely on an approach developed in~\cite{gruica2020common} where we prove that maximum rank-distance codes are sparse over large fields, with only very few exceptions. We refer to \cite{gruica2020common} for the proofs and include references when we do so.

In the second part of the paper we investigate the asymptotic density of subspace codes endowed with the injection metric. We establish the analogues of the results obtained in the Hamming metric and determine the decisive asymptotics for sparseness/density. As an application of our results in finite geometry, we study the density of partial spreads within the collection of $k$-subspaces having a prescribed cardinality.

\section{Preliminaries}

Throughout the paper, $q$ is a prime power, $\F_q$ is the finite field with $q$ elements, and $n \ge 2$ denotes an integer. We start by recalling some notions from classical coding theory; see e.g.~\cite{macwilliams1977theory}.

\begin{definition}
A (\emph{block}) \emph{code}  is a subset $\mC \subseteq \F_q^n$ of cardinality $|\mC| \ge 2$. The \emph{minimum} (\emph{Hamming}) \emph{distance} of $\mC$ is 
\begin{align*}
    \dH(\mC) = \min\{\dH(x,y) \mid x,y \in \mC, \; x \ne y\},
\end{align*}
where $\dH$ denotes the Hamming distance on $\F_q^n$. 
\end{definition}

\begin{remark}
Block codes in the Hamming metric can be defined over any alphabet of at least two symbols. In this paper, we only consider alphabet sizes that are equal to prime powers in order to treat block codes and subspace codes in a uniform way. All the results on block codes however extend to arbitrary alphabets.
\end{remark}

It is well-known that the cardinality of a code $\mC \subseteq \F_q^n$ with $\dH(\mC) \ge d$ satisfies $\log_q(|\mC|) \le n- d+1$; see \cite{singleton1964maximum}. This inequality is the famous \emph{Singleton bound} and codes meeting it with equality are called \emph{MDS} (\emph{Maximum Distance Separable}). 

Recall that the \emph{Hamming ball} of radius $0 \le r \le n$ and center $x \in \F_q^n$ is the set 
$\{y \in \F_q^n \mid \dH(x,y) \le r\}$. Its size is $$\sum_{i=0}^r \binom{n}{i}(q-1)^i$$ and does not depend on the center $x$. The (asymptotics of the) size of the Hamming ball will be of crucial importance in Section~\ref{sec:classicalcodes}.

It is natural to ask 
how the typical code of a given cardinality looks like in certain parameter ranges. In this paper, we concentrate on the scenario where the alphabet size $q$ is large.
To address this question formally, we consider the problem
of estimating 
the proportion of codes of 
minimum distance lower bounded by a given integer, say $d$, within the family of codes having the same cardinality. In order to simplify arguments in the sequel, we introduce the following terminology.

\begin{definition} \label{def:deltaclass}
For $1 \le d \le n$, let 
\begin{align*} 
    \deltaH(n,S,d) = \frac{|\{\mC \subseteq \F_q^n \mid |\mC|=S, \, \dH(\mC) \ge d\}|}{|\{\mC \subseteq \F_q^n \mid |\mC|=S \}|}
\end{align*}
denote the \emph{density function} of codes in $\F_q^n$ of cardinality $S$ and minimum distance at least $d$, among all codes in $\F_q^n$ of cardinality $S$.
\end{definition}

Since we focus on large alphabets, we study the asymptotics of the previous problem for $q$ going to infinity. More formally, we denote by $Q$ the set of prime powers, fix~$n$ and~$d$, and 
consider a sequence $(S_q)_{q \in Q}$ of integers with $S_q \ge 2$ for all $q \in Q$. We want to study how the asymptotic density $\lim_{q \to +\infty} \delta_q(n,S_q,d)$ depends on the asymptotics of the sequence $(S_q)_{q \in Q}$.

\begin{notation}
We use the Bachmann-Landau notation (``Little O'' and~``$\sim$'') to describe the asymptotic growth of real-valued functions defined on $Q$; see e.g.~\cite{de1981asymptotic}. We omit ``$q \in Q$'' when writing $q \to +\infty$ and often omit ``as $q \to +\infty$'' when writing, for example, ``$f(q) \in o(1)$''.
\end{notation}

In the second part of the paper, we will consider the same problem in the context of subspace codes endowed with the injection metric. Other distance functions can be studied with the same method. In this article, we only treat the Hamming and the injection distances because of space constraints.

\begin{definition}
For an integer $1 \le k \le n$, we denote by $\mG_q(k,n)$
the set of all $k$-dimensional subspaces of $\F_q^n$, also known as the \emph{Grassmannian}.
The (\emph{injection}) \emph{distance} between $X,Y \in \mG_q(k,n)$ is $\dI(X,Y) = k-\dim(X\cap Y)$ and the \emph{minimum} (\emph{injection}) \emph{distance} of a subspace code $\mC \subseteq \mG_q(k,n)$ is 
$$\dI(\mC)=\min\{\dI(X,Y) \mid X,Y \in \mC, \, X \neq Y\}.$$
\end{definition}

Note that the injection distance on $\mG_q(k,n)$ coincides with the so-called \textit{subspace distance} divided by 2.

\begin{remark} \label{rem:klen-k}
It is well-known and easy to see that taking orthogonals induces a one-to-one correspondence between subspace codes in $\mG_q(k,n)$ of cardinality $S$ and minimum distance $d$ and subspace codes in $\mG_q(n-k,n)$ of cardinality $S$ and minimum distance $d$; see e.g.~\cite[Section III]{koetter2008coding}. Therefore, to simplify certain statements throughout the paper, we will always assume $k \le n-k$ in the sequel. 
\end{remark}

As for the Hamming distance, we will need to consider the ball of a certain radius in the metric space $(\mG_q(k,n),\dI)$.
For a radius $0 \le r \le k$, the latter is the set
\smash{$\{Y \in \mG_q(k,n) \mid \dI(X,Y) \le r\}.$} Its cardinality
can be conveniently expressed in terms of the $q$-ary binomial coefficient, which 
counts the number 
of $\ell$-dimensional subspaces of an $m$-dimensional space over $\F_q$ and is denoted by  $$\qbin{m}{\ell}{q}.$$ 

\begin{proposition}[\text{see \cite[Theorem 5]{koetter2008coding}}]  \label{prop:ballsizeI}
For all $X \in \mG_q(k,n)$ we have
\begin{align*}
    |\{Y \in \mG_q(k,n) \mid \dI(X,Y) \le r\}| = \sum_{i=0}^{r} q^{i^2}\qbin{k}{i}{q} \qbin{n-k}{i}{q}.
\end{align*}
In particular, the size of the injection ball in $\mG_q(k,n)$ does not depend on the choice of its center.
\end{proposition}

The following is the analogue of the Singleton bound in the context of subspace codes.

\begin{theorem} [\text{see \cite[Theorem 9]{koetter2008coding}}] \label{thm:singletonbound}
Suppose $k \le n-k$. For any subspace code \smash{$\mC \subseteq \mG_q(k,n)$} with minimum distance $\dI(\mC) \ge d$ we have
\begin{align*}
    |\mC| \le \qbin{n-d+1}{n-k}{q}.
\end{align*}
\end{theorem}

While computing the largest size of a subspace code of given minimum distance is an open problem, a family of asymptotically optimal subspace codes has been constructed in \cite{koetter2008coding}. These are the so-called \emph{Reed-Solomon-like codes}.

\begin{theorem}[\text{see \cite[Section V]{koetter2008coding}}]
For $k \le n-k$ and $1 \le d \le k$, there exists a subspace code $\mC \subseteq \mG_q(k,n)$ with $\dI(\mC)=d$ and $|\mC|=q^{(n-k)(k-d+1)}$. 
\end{theorem}

Note that in this paper we only concentrate on the Singleton-type bound because it is asymptotically sharp for $q$ large. Several other bounds on the size of subspace codes are available~\cite{xia2009johnson, etzion2011error}.

The following is the analogue of Definition~\ref{def:deltaclass} for subspace codes.

\begin{definition} \label{def:deltasub}
For $1 \le d \le k$, let 
\begin{align*}
\deltaI(n,k,S,d) = \frac{|\{\mC \subseteq \mG_q(k,n)\mid |\mC|=S, \, \dI(\mC) \ge d\}|}{|\{\mC \subseteq \mG_q(k,n) \mid |\mC|=S \}|}
\end{align*}
denote the \emph{density function} of subspace codes  in $\mG_q(k,n)$ of cardinality $S$ and minimum distance at least $d$, among all subspace codes in $\mG_q(k,n)$ of cardinality $S$.
\end{definition}

When studying the typical subspace code, we will fix the three parameters $(n,k,d)$
and 
consider a sequence $(S_q)_{q \in Q}$ of integers with $S_q \ge 2$ for all $q \in Q$. We will then study how the asymptotic density $\lim_{q \to +\infty} \delta_q(n,k,S_q,d)$ depends on the asymptotics of  $(S_q)_{q \in Q}$ for $q$ large.

\section{Graph Theory Tools} \label{sec:bounds}
In this section we briefly state some graph theory tools we will need later. The results are taken from~\cite{gruica2020common} and the proofs are omitted.

\begin{definition}
A (\emph{directed}) \emph{bipartite graph} is a 3-tuple $\mB=(\mV,\mW,\mE)$, where $\mV$ and $\mW$ are finite non-empty sets and $\mE \subseteq \mV \times \mW$. The elements of $\mV \cup \mW$ are the \emph{vertices} of the graph. We say that a vertex~$W \in \mW$ is
\emph{isolated} if there is no $X \in \mV$ with $(X,W) \in \mE$. We say that
$\mB$ is 
\emph{left-regular} of \emph{degree} $\partial \ge 0$ if for all $X \in \mV$
$$|\{W \in \mW \mid (X,W) \in \mE\}| = \partial.$$
\end{definition}

In order to give bounds for the number of non-isolated vertices in a bipartite graph, we need the notion of an association.

\begin{definition} \label{def:assoc}
Let $\mV$ be a finite non-empty set and let $r \ge 0$ be an integer. An \emph{association} on $\mV$ of \emph{magnitude} $r$ is a function 
$\alpha: \mV \times \mV \to \{0,...,r\}$ satisfying the following:
\begin{itemize}
\item[(i)] $\alpha(X,X)=r$ for all $X \in \mV$;
\item[(ii)] $\alpha(X,Y)=\alpha(Y,X)$ for all $X,Y \in \mV$.
\end{itemize}
\end{definition}

Let $\mB=(\mV,\mW,\mE)$ be a finite bipartite graph and let $\alpha$ be an association on~$\mV$ of magnitude $r$.  We say that $\mB$ is \emph{$\alpha$-regular} if for all  $(X,Y) \in \mV \times \mV$ the number of vertices $W \in \mW$ with $(X,W) \in \mE$ and 
$(Y,W) \in \mE$ only depends on $\alpha(X,Y)$. If this is the case, we denote this number by~$\mW_\ell(\alpha)$, where $\ell=\alpha(X,Y)$. 

\begin{remark}
Note that an $\alpha$-regular bipartite graph for an association $\alpha$ is necessarily left-regular of degree $\partial=\mW_r(\alpha)$.
\end{remark}

The main results stated in this paper will be derived by the following two bounds.

\begin{lemma}[\text{see \cite[Lemmma 3.2]{gruica2020common}}] \label{lem:upperbound}
Let $\mB=(\mV,\mW,\mE)$ be a bipartite and left-regular graph of degree $\partial>0$.
Let $\mF \subseteq \mW$ be the collection of non-isolated vertices of $\mW$.
We have
$$|\mF| \le |\mV| \, \partial.$$
\end{lemma}

\begin{lemma}[\text{see \cite[Lemmma 3.5]{gruica2020common}}] \label{lem:lowerbound}
Let $\mB=(\mV,\mW,\mE)$ be a finite bipartite $\alpha$-regular graph, where $\alpha$ is an association on~$\mV$ of magnitude~$r$. Let $\mF \subseteq \mW$ be the collection of non-isolated vertices of $\mW$. If
$\mW_r(\alpha) >0$, then 
$$|\mF| \ge  \frac{\mW_r(\alpha)^2 \, |\mV|^2}{\sum_{\ell=0}^r  \mW_\ell(\alpha) \, |\alpha^{-1}(\ell)|}.$$
\end{lemma}

The previous lemma follows by combining the notion of an association and the Cauchy-Schwarz Inequality. We refer to~\cite{gruica2020common} for the complete proof.

\section{The Typical Block Code} \label{sec:classicalcodes}

We show how to apply the results of Section~\ref{sec:bounds} to derive estimates for the number of codes in the Hamming metric having minimum distance bounded from below.

\begin{notation}
In this section, let $n$ and $d$ be fixed integers with $2 \le d \le n$ and let $(S_q)_{q \in Q}$ be a sequence of integers with $S_q \ge 2$ for all $q \in Q$ and  for which $\lim_{q \to +\infty} S_q$ exists.
We work with the bipartite graphs $$\mB^\HH_q=(\mV^\HH_q,\mW^\HH_q,\mE^\HH_q),$$ where  $\mV^\HH_q=\{\{x,y\} \subseteq \F_q^n \mid x \ne y, \, \dH(x, y) \le d-1 \}$, $\mW^\HH_q$ is the collection of codes in $\F_q^n$ of cardinality $S_q$, and 
 $(\{x,y\},\mC) \in \mE^\HH_q$ if and only if $\{x,y\} \subseteq \mC$.
\end{notation}

From now on, let $\BallH_q$ denote the size of the Hamming ball in $\F_q^n$ of radius $d-1$. We have \begin{align*}
    |\mV^\HH_q| = \frac{1}{2} q^n \left(\BallH_q-1\right), \quad  |\mW^\HH_q| = \binomi{q^n}{S_q}.
\end{align*}
It follows from the definitions that  $\mB^\HH_q$ is a left-regular graph of degree $$\binomi{q^n-2}{S_q-2}.$$ 
Therefore, by applying Lemma~\ref{lem:upperbound} we obtain the following result.

\begin{theorem} \label{thm:hamupp}
Let $\mF^\HH_q$ be the collection of codes $\mC \subseteq\F_q^n$ that have cardinality $S_q$ and minimum Hamming distance at most~$d-1$. For all $q \in Q$ we have
\begin{align*}
    |\mF^\HH_q| \le \frac{1}{2} q^n \left( \BallH_q-1\right)\binomi{q^n-2}{S_q-2}.
\end{align*}
\end{theorem}

We now use Lemma~\ref{lem:lowerbound} to derive a lower bound for the number of codes of minimum distance bounded from above.

\begin{theorem} \label{thm:hamlow}
Let $\mF^\HH_q$ be the collection of codes $\mC \subseteq\F_q^n$ that have cardinality $S_q$ and minimum Hamming distance at most $d-1$. Define the quantities
\begin{align*}
    \beta_q^\HH(0) &= \frac{1}{2}q^n (\BallH_q-1) - 2\BallH_q +3, \\
    \beta_q^\HH(1) &= 2\BallH_q-4, \\
    \Omega^\HH_q &= 1 + \beta_q^\HH(1) \, \frac{S_q-2}{q^n-2} + \beta_q^\HH(0) \,  \frac{(S_q-2)(S_q-3)}{(q^n-2)(q^n-3)}.
\end{align*}
For all $q \in Q$ we have
\begin{align*}  
    |\mF^\HH_q| \ge \frac{\displaystyle q^n (\BallH_q-1) \binom{q^n-2}{S_q-2}}{2\Omega^\HH_q}.
\end{align*}
\end{theorem}
\begin{proof}
Let $\alpha: \mV^\HH_q \times \mV^\HH_q \longrightarrow \{0,1,2\}$ be defined
by $$\alpha(\{x,y\},\{t,z\}) := 4-|\{x,y,t,z\}|$$
for all $x,y,t,z \in \F_q^n.$
\begin{claim} \label{cl:A}
For all $q \in Q$ we have
\begin{align*}
|\alpha^{-1}(2)| &= |\mV^\HH_q|,      \\
|\alpha^{-1}(1)| &= 2|\mV^\HH_q|(\BallH_q-2),  \\
|\alpha^{-1}(0)| &= |\mV^\HH_q|(|\mV^\HH_q|-2\BallH_q+3).
\end{align*}
\end{claim}
\begin{clproof}
It is easy to see that $|\alpha^{-1}(2)| = |\mV^\HH_q|$. 
The elements of $\alpha^{-1}(1)$ can all be constructed by freely choosing 
 $\{x,y\} \in \mV^\HH_q$
 and then $\{z,t\} \in \mV^\HH_q$ with 
 either $t=x$ or $t=y$ and $$z \in \{v \in \F_q^n \mid \dH(t,v) \le d-1\} \backslash \{x,y\}.$$  Therefore $$ |\alpha^{-1}(1)|=2|\mV^\HH_q|(\BallH_q-2).$$
To compute $|\alpha^{-1}(0)|$ we simply note that  $$|\mV^\HH_q|^2=|\alpha^{-1}(0)|+|\alpha^{-1}(1)|+|\alpha^{-1}(2)|.$$ Therefore the value of $|\alpha^{-1}(0)|$ 
follows from the values of $|\alpha^{-1}(1)|$ and $|\alpha^{-1}(2)|$.
\end{clproof}

One easily checks that $\alpha$ is an association on $\mV^\HH_q$ and that 
the bipartite graph $\mB^\HH_q$ is regular with respect to $\alpha$. More precisely, for $(\{x,y\},\{t,z\}) \in \mV^\HH_q \times \mV^\HH_q$ let $\ell=\alpha(\{x,y\},\{t,z\})$. Then 
\begin{align} \label{eq:walpha}
    \mW^\HH_{q,\ell}(\alpha) := |\{W \in \mW^\HH_q \mid \{x,y,t,z\} \subseteq W\}| \;=  \binomi{q^n-4+\ell}{s-4+\ell}.
\end{align}
We can now apply Lemma~\ref{lem:lowerbound} obtaining that $|\mF_q^\HH|$ is lower bounded by
\begin{equation*}
\frac{\mW^\HH_{q,2}(\alpha)^2 \, |\mV^\HH_q|^2}{|\alpha^{-1}(2)|\mW^\HH_{q,2}(\alpha) + |\alpha^{-1}(1)|\mW^\HH_{q,1}(\alpha) + |\alpha^{-1}(0)|\mW^\HH_{q,0}(\alpha)}.
\end{equation*}
Finally, combining the identity
\begin{align} \label{eq:binomi}
\binom{m}{\ell} = \frac{m}{\ell} \binom{m-1}{\ell-1}
\end{align}
with the formulas in Claim~\ref{cl:A} and~\eqref{eq:walpha}, easy computations yield the desired result.
\end{proof}

We can use the previous two results to study the asymptotic density of non-linear codes of given size and minimum distance bounded from below. 

\begin{corollary} \label{cor:mainclasscor}
For all $q \in Q$ we have
\begin{align*}
    \deltaH(n,S_q,d) &\ge 1- \frac{(\BallH_q-1)S_q(S_q-1)}{2\left(q^n-1\right)}, \\
    \deltaH(n,S_q,d) &\le 1-\frac{(\BallH_q-1)S_q(S_q-1)}{2\Omega^\HH_q(q^n-1)}, 
\end{align*}
where $\Omega^\HH_q$ is defined as in Theorem~\ref{thm:hamlow}.
\end{corollary}

It is now interesting to study the asymptotics of the previous bounds as the alphabet size $q$ tends to infinity. For this,
 we will need the following estimate:
\begin{align}
\BallH_q \sim \binomi{n}{d-1}q^{d-1} \qquad \textnormal{
 as $q \to +\infty$}.
\end{align}
We are now ready to state and prove one of the main results of this paper.

\begin{theorem} \label{thm:classasymp}
Let $\gamma_q^\HH=\sqrt{q^{n-d+1}}$. We have 
\begin{align*}
    \lim_{q \to +\infty}\delta_q^\HH(n,S_q,d)= \begin{cases} 1 \quad &\textnormal{if $S_q \in o(\gamma_q^\HH)$,} \\
    0 \quad &\textnormal{if $\gamma_q^\HH \in o(S_q)$.}
    \end{cases}
\end{align*}
\end{theorem}
\begin{proof}
If $S_q \in o(\gamma_q^\HH)$,
then 
\begin{align} \label{eq:lowerboundclass}
\lim_{q \to +\infty} \frac{(\BallH_q-1)S_q(S_q-1)}{2\left(q^n-1\right)} =0.\end{align}
Along with Corollary~\ref{cor:mainclasscor}, this yields the first limit in the statement. Now suppose that
$\gamma_q^\HH \in o(S_q)$ as $q \to +\infty$.
 Then it is not hard to see that  $$\Omega_q^\HH \sim \frac{\binomi{n}{d-1}S_q^2}{2 (\gamma_q^{\HH})^2} \quad  \textnormal{as $q \to +\infty$}.$$
 Therefore, since $1 \in o(S_q)$ as $q \to +\infty$, we have
 $$\frac{(\BallH_q-1)S_q(S_q-1)}{2\Omega^\HH_q(q^n-1)} \sim 1 \quad \textnormal{ as $q \to +\infty$}.$$
 This gives the second limit in the statement thanks to Corollary~\ref{cor:mainclasscor}.
\end{proof}

\begin{remark}
The Gilbert-Varshamov bound (see \cite{gilbert1952comparison,varshamov1957estimate}) can be used to show the existence of error-correcting codes $\mC \subseteq \F_q^n$ having minimum distance lower bounded by $d \ge 2$ and cardinality $$\sim \frac{q^{n-d+1}}{\binom{n}{d-1}} \mbox{ as $q \to +\infty$.}$$ Our results show that, while these codes exist, they are very far from being dense. In particular, over a sufficiently large alphabet, the typical non-linear code whose cardinality is close to the Gilbert-Varshamov bound computed for a given distance $d$, has minimum distance strictly smaller than $d$. 
\end{remark}

\section{The Typical Subspace Code} \label{sec:subspacecodes}

In this section we establish the analogue of Theorem~\ref{thm:classasymp} for subspace codes endowed with the injection metric.

\begin{notation}
We fix integers $n$, $k$ and $d$ with $2 \le  d \le k\le n-k$ and let $(S_q)_{q \in Q}$ be a sequence of integers with $S_q \ge 2$ for all $q \in Q$ and  for which $\lim_{q \to +\infty} S_q$ exists. We consider the bipartite graphs $$\mB^\II_q=(\mV^\II_q,\mW^\II_q,\mE^\II_q),$$ where $\mV^\II_q=\{\{X,Y\} \subseteq \mG_q(k,n) \mid X \ne Y, \, \dI(X, Y) \le d-1 \},$ $\mW^\II_q$ is the collection of subspace codes in $\mG_q(k,n)$ having cardinality $S_q$, and $(\{X,Y\},\mC) \in \mE^\II_q$ if and only if $\{X,Y\} \subseteq \mC$. 
\end{notation}

From now on, let $\BallI_q$ denote the size of the injection ball in $\mG_q(k,n)$ of radius $d-1$ given in Proposition~\ref{prop:ballsizeI}. 
We have \begin{align*}
    |\mV^\II_q| = \frac{1}{2}\qbin{n}{k}{q} \left(\BallI_q-1\right), \quad  |\mW^\II_q| = \binomi{\qbin{n}{k}{q}}{S_q}.
\end{align*}
By applying Lemma~\ref{lem:lowerbound} we obtain the following upper bound.

\begin{theorem} \label{thm:subupp}
Let $\mF_q^\II$ be the collection of subspace codes $\mC \subseteq \mG_q(k,n)$ of cardinality $S_q$ and minimum injection distance at most $d-1$. 
We have
$$|\mF_q^\II| \le  \frac{1}{2}\qbin{n}{k}{q} \left(\BallI_q-1\right) \binomi{\qbin{n}{k}{q}-2}{S_q-2}.$$
\end{theorem}

We then proceed as we did for block codes, obtaining the analogue of Theorem~\ref{thm:hamlow}. The proof is omitted.

\begin{theorem} \label{thm:sublow}
Let $\mF_q^\II$ be the collection of subspace codes $\mC \subseteq \mG_q(k,n)$ of cardinality $S_q$ and minimum injection distance at most $d-1$. Define the quantities
\begin{align*}
    \beta_q^\II(0) &= \frac{1}{2}\qbin{n}{k}{q} (\BallI_q-1) - 2\BallI_q +3, \\
    \beta_q^\II(1) &= 2\BallI_q-4, \\
    \Omega^\II_q &= 1 + \beta_q^\II(1)  \frac{S_q-2}{\qbin{n}{k}{q}-2} \, + \, \beta_q^\II(0)   \frac{(S_q-2)(S_q-3)}{\left(\qbin{n}{k}{q}-2\right)\left(\qbin{n}{k}{q}-3\right)}.
\end{align*}
We have
\begin{align*}
|\mF^\II_q| \ge \frac{\displaystyle \qbin{n}{k}{q} (\BallI_q-1) \binomi{\qbin{n}{k}{q}-2}{S_q-2}}{2\Omega^\II_q}.
\end{align*}
\end{theorem}

From Theorems~\ref{thm:subupp} and~\ref{thm:sublow} we derive bounds on the density functions of subspace codes in~$\mG_q(k,n)$ of cardinality~$S_q$ and minimum distance at least $d$.

\begin{corollary} \label{cor:denssub}
Let $\Omega^\II_q$ be defined as in Theorem~\ref{thm:sublow}. For all $q \in Q$ we have
\begin{align*}
    \deltaI(n,k,S_q,d) &\ge 
    1-\frac{(\BallI_q-1)S_q(S_q-1)}{2\left(\qbin{n}{k}{q}-1\right)}, \\
    \deltaI(n,k,S_q,d) &\le
    1-\frac{(\BallI_q-1)S_q(S_q-1)}{2\Omega^\II_q\left(\qbin{n}{k}{q}-1\right)}.
\end{align*}
\end{corollary}

In order to compute the asymptotic density of subspace codes, we will need the following result.

\begin{proposition} \label{prop:asball}
We have $\BallI_q \sim q^{(d-1)(n-d+1)}$ as $q \to + \infty$.
\end{proposition}

The proof of Proposition~\ref{prop:asball} follows from well-known estimates for $q$-ary binomial coefficients and is therefore omitted.

\begin{theorem} \label{thm:subasymp}
Let $\gamma_q^\II=\sqrt{q^{k(n-k)-(d-1)(n-d+1)}}$. We have
\begin{align*}
    \lim_{q \to +\infty}\delta_q^\II(n,k,S_q,d)= \begin{cases} 1 \quad &\textnormal{if $S_q \in o(\gamma_q^\II)$,} \\
    0 \quad &\textnormal{if $\gamma_q^\II \in o(S_q)$.}
    \end{cases}
\end{align*}
\end{theorem}
\begin{proof}
If $S_q \in o(\gamma_q^\II)$ as $q \to+\infty$, then one can check easily that
\begin{align} \label{eq:subone}
    \lim_{q \to +\infty} \frac{(\BallI_q-1)S_q(S_q-1)}{2\left(\qbin{n}{k}{q}-1\right)} =0.
\end{align}
This computes the first limit in the statement thanks to Corollary~\ref{cor:denssub}.
If $\gamma_q^\II \in o(S_q)$ as $q \to +\infty$, then $$\Omega_q^\II \sim \frac{S_q^2}{2{(\gamma_q^\II)^2}} \quad  \textnormal{as $q \to +\infty$.}$$
Therefore
\begin{align*}
    \frac{(\BallI_q-1)S_q(S_q-1)}{2\Omega^\II_q\left(\qbin{n}{k}{q}-1\right)} \sim 1 \quad \textnormal{ as $q \to +\infty$,}
\end{align*}
concluding the proof by Corollary~\ref{cor:denssub}.
\end{proof}

\section{Partial Spreads in Finite Geometry}

The results of the previous section have a curious application in finite geometry. A subspace code $\mC \subseteq \mG_q(k,n)$ with injection distance $k$ is a so-called \emph{partial spread}. Therefore, our results tell us for which cardinalities $S$ a uniformly random collection of $S$ subspaces in $\mG_q(k,n)$ form a partial spread with high probability (for $q$ large). The decisive cardinality for sparseness/density, in the asymptotics, is
$$S_q\sim \sqrt{q^{n-2k+1}}.$$
More precisely, the following holds.

\begin{corollary} \label{cor:partialspr}
Let $n \ge k \ge 1$ be integers. We have
\begin{align*}
    \lim_{q \to +\infty}\delta_q^\II(n,k,S_q,k)= \begin{cases} 1 \quad &\textnormal{ if $S_q  \in o \left(\sqrt{q^{n-2k+1}} \right)$,} \\
    0 \quad &\textnormal{ if $\sqrt{q^{n-2k+1}}  \in o(S_q)$.}
    \end{cases}
\end{align*}
In particular, assume that $k$ divides $n$ and let $\mC$ be a uniformly random collection of $(q^n-1)/(q^k-1)$ $k$-subspaces of $\F_q^n$. The probability that $\mC$ is a spread goes to 0 as $q$ tends to infinity.
\end{corollary}

\section{Discussion and Future Work}

We described the behavior of the typical non-linear code in the Hamming and the injection metric. In the Hamming metric setting, the typical non-linear code of minimum distance at least $d$ and large cardinality is far from being MDS. This is in strong contrast with the behavior of linear block codes in the Hamming metric. 

An an application of our results in projective geometry, we determine the (asymptotics of) cardinalities $S$ for which a uniformly random collection of subspaces form a partial spread with high probability over a large field. In particular, we conclude that spreads are very rare objects.

A natural problem inspired by the above results is that of understanding which structural invariants of a metric space  determine the distance properties of a uniformly random subset.
The approach taken in this paper shows that graph theory is a valid tool for understanding these structural invariants. There are very natural information theory questions connected to these problems, which will be explored in future work.

\bibliographystyle{amsplain}
\bibliography{ourbib}
\end{document}